\documentclass{llncs}

\usepackage{color}

\usepackage{amsmath, amssymb}
\usepackage{graphicx}

\usepackage{xspace}
\usepackage{ifthen}
\usepackage{paralist} 

\usepackage[dvipsnames]{xcolor} 
\usepackage{changebar} 
\newcommand{\hint}[2]{} 

\newcommand{\FixedGap}[1][]{\ifthenelse{\equal{#1}{}}{\textsc{2-Gap}}{\ensuremath{\textsc{Gap}_{#1}}}\xspace}

\title{Exact VC-dimension for $L_1$-visibility\\ of points in simple polygons}
\titlerunning{VC-dimension for $L_1$-visibility} 
\author{Elmar Langetepe
\and
Simone Lehmann
}
\authorrunning{Langetepe et al.} 
\tocauthor{}
\institute{
University of Bonn, Department of Computer Science,
D-53113 Bonn, Germany}

\begin{document}

\maketitle
\begin{abstract}
The VC-dimension plays an important role for the 
algorithmic problem of guarding art galleries efficiently. 
We prove that inside a simple polygon 
at most~$5$ points can be shattered by $L_1$-visibility polygons and 
give an example where 5 points are shattered. 
 The VC-dimension is exactly~$5$. The proof idea for the upper bound 
 is different from previous approaches.
\end{abstract}

\section{Introduction and Definitions}\label{intro-sect}

In this paper we study a visibility problem that is related to efficient algorithmic solutions of the \emph{art gallery problem}.  
Such problems have a long tradition, for example 
one can ask for the minimum set of guards so that the union of visibility
regions covers a simple polygon $P$; see \cite{g-aaagp-10,o-agta-87}. 

The classic $\epsilon$-net Theorem implies that $O(d\;r\log\; r)$ many stationary guards with $360^\circ$ vision are sufficient to cover $P$, provided that each point in $P$ sees at least an $1/r$-th part of the area of $P$. The constant hidden in O is very close to~$1$; see \cite{km-ggwep-97,kpw-atbfe-92}. Here the (constant) $d$ denotes  the 
well-known $VC$-dimension for visibility polygons of points in 
simple polygons. If $d$ is small, only few guards are required. 
The definition of $VC$-dimension goes back to {\bf V}apnik and {\bf C}hervonenkis; see \cite{vc-ucrfe-71}. 
Note that for computing the number of guards required there are also direct approaches that do not make use of this theory.  
Kirkpatrick~\cite{k-ggh-00} obtained an $64\; r \log \log r$ upper bound to the number of (boundary) guards needed to cover the boundary of $P$. 
This was further examined in~\cite{kk-iagsg-11}.

We briefly explain the concept of VC-dimension for visibility polygons 
of points in simple polygons. For $L_2$-visibility two points $p$ and $q$  inside
$P$ are \emph{visible} or \emph{see} each other, 
if the line-segment $pq$ fully lies inside $P$. 
Given a simple polygon $P$ and a finite set $S=\{p_1,p_2,\ldots,p_n\}$ of points in $P$, we say that a subset $T\subseteq S$ can be \emph{shattered} in $P$, if there exists a \emph{viewpoint} $v_T\in P$ such that $v_T$ exactly \emph{sees} all points in $T$ but definitely sees no point in 
$S\setminus T$. If such a viewpoint $v_T$ (or a 
set $V(T)\subseteq P$ of such viewpoints) for any of the $2^n$ subsets $T\subseteq S$ exists, we say that the whole set $S$ can be \emph{shattered}. The VC-dimension $d$ is the maximum cardinality of a set $S$ such that a polygon $P$ exists where \emph{all} subsets $T$ of $S$ can be shattered. 

The VC-dimension is also used in other computational areas. 
 In computational learning theory the
 use of VC-dimension helps for deriving upper and lower bounds on the number of necessary training  examples; see~\cite{kv-iclt-94}.

Some work has been done on the VC-dimension of $L_2$-visibility  in simple polygons.  In~\cite{v-ggwnp-98} $d \in[6,23]$ was shown; compare also~\cite{m-ldg-02}. Figure~\ref{lowerboundL2-fig} shows 
the best known lower bound for~$6$  points that can be shattered. 
At WADS 2009~\cite{gk-nrvsp-09} it was shown that 14 points on the boundary of a Jordan curve cannot be shattered. 
This upper bound was further generalized to $d\leq 14$ for arbitrary point 
in \cite{gk-nubvv-14} . So the current known interval for~$d$ is~$[6,14]$. 
It is an open conjecture that the VC-dimension is exactly~$6$. 
An upper bound of~$6$ was shown for point sets on the boundary  of monotone polygons in~\cite{gkw-vvmp-14} and there are some 
results for external visibility~\cite{ikdv-vev-04}.


In this paper we exactly answer the VC-dimension question for $L_1$-visibility of 
point sets in simple polygons. 
For a point $p\in P$ the $L_1$-visibility polygon of $p$ (the set of all points seen 
from $p$) is always larger than the $L_2$-visibility polygon of $p$. 
Note that the notion of VC-dimension is related to 
the property of \emph{seeing} points but also to the fact of \emph{not-seeing} other points. So there is no direct relationship between
$L_1$- and $L_2$-visibility. 

The proof idea for the upper bound  used here is different from the 
previous results. This is interesting in its own right. 
We show that the subset, $V(S)$, of $P$ that sees all points of $S$ is always path 
connected. Furthermore, the areas of the subsets of $P$ that misses exactly one 
point, say $V(S\setminus\!\!\{p_i\})$, have  a common boundary with $V(S)$. 
This means that any $V(S\setminus\!\!\{p_i\})$ is located along the boundary of $V(S)$. Interestingly, this is independent from 
$L_1$- or $L_2$-visibility. 
For $L_1$-visibility a simple argument already says that 
only 8 such regions around $V(S)$ can exist and no more 
than 8 points can be shattered. 
But we can further lower down the  number of potential areas $V(S\setminus\!\!\{p_i\})$ located 
around 
$V(S)$ to $5$ by considering sets $V(S\setminus\!\!\{p_i,p_j\})$ 
for two sets $V(S\setminus\!\!\{p_i\})$ and $V(S\setminus\!\!\{p_j\})$.
The cardinality~$5$ coincidence with our lower bound example and 
the VC-dimension is exactly~$5$. 

This also means that we even show a slightly stronger result. 
For $L_1$-visibility inside a simple polygon and for a set $S=\{p_1,p_2,\ldots,p_n\}$ of points we can 
shatter all subsets $S\!\!\setminus\{p_i\}$ and  all subsets $S\!\!\setminus\{p_i,p_j\}$ and the set $S$ from $P$ for 
no more than~$n=5$ points and there is an example where this
subset-shattering for $n=5$ is possible. 

In this work the Figures~\ref{lowerbound-fig}, \ref{notvisible-fig} and \ref{lowerboundL2-fig} were generated 
by a visualisation-tool from~\cite{ps-vgp-15}, the corresponding 
software project was supervised by the first author. 

\section{Definitions for $L_1$-visibility}\label{prelim-sect}

For a simple polygon $P$ we define $L_1$-visibility and $L_1$-cuts associated to vertices 
(or axis-parallel edges).
\begin{figure*}
\begin{center}
\includegraphics[scale=0.5]{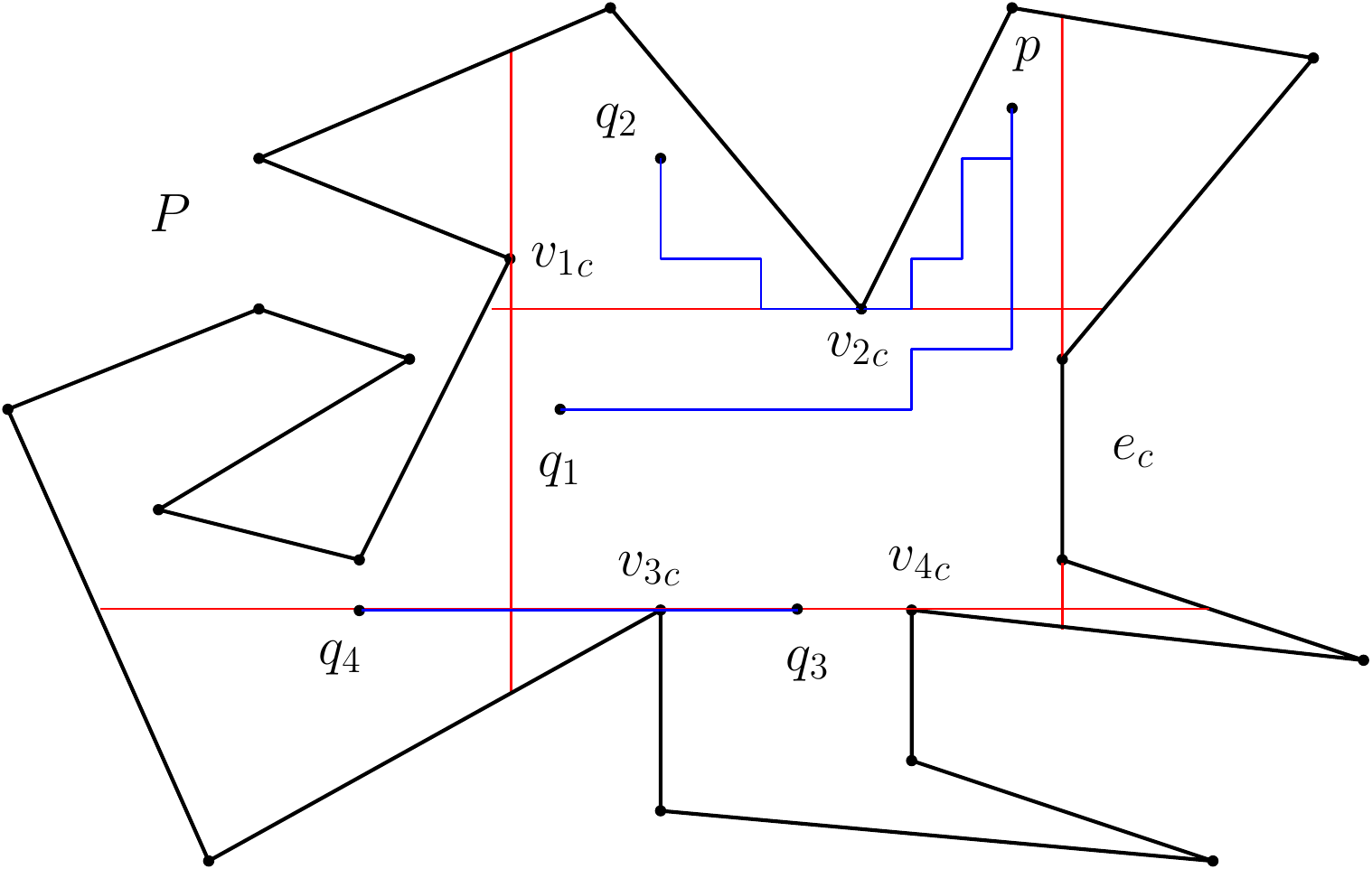}
\caption{The points $p$ and $q_1$ are $L_1$-visible whereas 
$p$ and $q_2$ are not $L_1$-visible because the $L_1$-visibility is 
blocked by the horizontal $L_1$-cut of the locally $Y$-minimal vertex $v_2$.
The vertex itself does not block the visibility along the cut, for example 
$q_3$ and $q_4$ are visible. 
 The axis-parallel locally $X$-minimal edge $e_c$ analogously defines a
vertical $L_1$-cut. With respect to directions the $L_1$-cut 
of ${v_2}_c$ can be labeled by N (north) whereas the $L_1$-cut 
of ${v_1}_c$ is labeled by W (west). Only four directions 
are given. In non-general position a cut can be evoked by different vertices (or edges)
${v_3}_c$ and ${v_4}_c$. 
}
\label{DefinitionL1-fig}
\end{center}
\end{figure*}
Consider two points  $p$ and $q$ inside $P$ 
as given in Figure~
\ref{DefinitionL1-fig}.  If a shortest $L_1$-path between $p$ and $q$ inside $P$ is $X$- \emph{and} $Y$-monotone, 
$p$ and $q$ are denoted as $L_1$-visible inside $P$.
The $L_1$-visibility between two points in $P$ can be blocked 
by  axis-parallel cuts emanating from locally $X$- or $Y$-maximal 
(or  locally $X$- or $Y$-minimal) vertices $v_c$ along the boundary of $P$; 
see Figure~\ref{DefinitionL1-fig} for some examples. 
For such a locally minimal or maximal vertex $v_c$, the axis-parallel cut
emanates in both directions until it hits the boundary. 
If  $v_c$  is minimal or maximal in $Y$-direction, the corresponding 
cut is horizontal,  if $v_c$  is minimal or maximal in $X$-direction, 
the $L_1$-cut is vertical. If $P$ is in general position 
such $L_1$-cuts subdivide the polygon into three disjoint parts. 
If $P$ is allowed to have axis-parallel edges, analogously an $L_1$-cuts  emanate in both directions from a corresponding egde $e_c$. 
Both vertices of $e_c$ are locally maximal or minimal. The cut is 
associated to the edge $e_c$. 

In this paper for convenience we  make use of an general position assumption for the polygon which says that no three vertices are on the same line and two vertices have the same $X$-or $Y$-coordinate, if and only if they share an edge. So we allow axis-parallel edges. 
Please note that all arguments also hold for non-general position. 
In this case an $L_1$-cut can be evoked by different vertices (or edges), see 
${v_3}_c$ and ${v_4}_c$ in Figure~\ref{DefinitionL1-fig}. 
For maintaining the arguments it is sufficient to  associated the cut to a single vertex. 
It is allowed to change this vertex, if this is necessary. 

With respect to directions, there can be at most four different 
kinds of $L_1$-cuts, depending on $X$- and $Y$-maximality 
or $X$- and $Y$-minimality of the corresponding vertex (or edge). For convenience we label the cuts by
the direction $\{N,E,S,W\}$, where the label means that 
the corresponding vertex (edge) \emph{lies} in this direction. 
For example an $L_1$-cut imposed by a  locally 
$Y$-minimal vertex is label by $N$ (north) and so on; 
see also Figure~\ref{DefinitionL1-fig}.

\section{Lower bound on the VC-dimension}\label{lowerbound-sect}

The lower bound of $5$ is shown by the example given in 
Figure~\ref{lowerbound-fig}. 
Note that the corresponding 
polygon $P$ need not be axis-parallel. The colors of the regions 
inside $P$ indicate the number of points 
that are shattered (red=5,brown=4, light-green=3 and so on).
In order to not overload the figure not all areas are labeled with the 
subset of points that are shattered in $P$. 
The reader can simply check that any of the $2^5=32$ subsets $T$ of $S=\{1,2,3,4,5\}$ 
is shattered by all points in some area $V(T)$ in~$P$. 
\begin{figure*}
\begin{center}
\includegraphics[scale=0.4,viewport = 35 35 500 500]{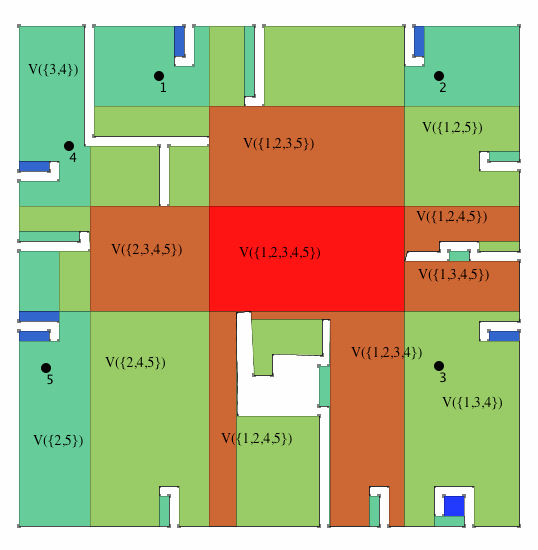}
\caption{Five points that are shattered by $L_1$-visibility polygons 
inside a simple polygon. The colors indicate the number of points 
that are shattered (red=5,brown=4, light-green=3 and so on).
Some regions are labeled by the point set that is precisely visible. 
Altogether, $2^5=32$ disjoint areas are required.  Note that the polygon need not be 
axis-parallel. }
\label{lowerbound-fig}
\end{center}
\end{figure*}
%

\section{Upper bound on the VC-dimension}\label{upperbound-sect}

Let us assume that inside a simple polygon $P$ a set $S:=\{p_1,p_2,\ldots,p_n\}$  of $n$ points can be shattered by $L_1$-visibility polygons. 
For a subset $T\subseteq S$ let $V(T)\subset P$ 
denote the union of all points in $P$ which 
sees all  points of $T$ but  no point of $S\setminus T$. 
Consider the set $V(S)\subset P$ that sees all points of $P$ (red areas in the examples  of Figure~\ref{lowerbound-fig} and Figure~\ref{notvisible-fig}).
\begin{figure*}
\begin{center}
\includegraphics[scale=0.4,viewport = 50 50 500 500]{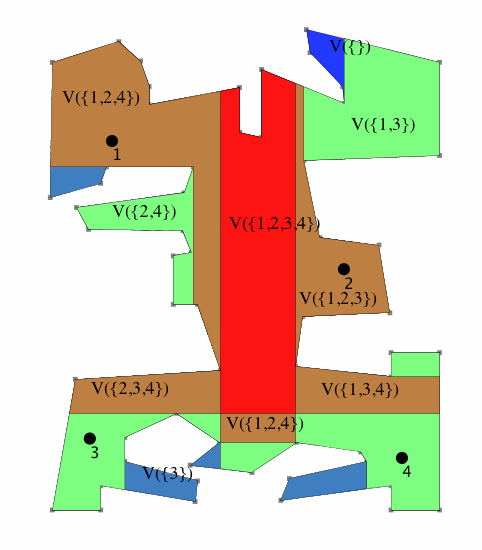}
\caption{Four points $S=\{1,2,3,4\}$ inside $P$ are shattered by 
$L_1$-visibility polygons. 
The union, $V(\{1,2,3,4\})$, of all points in $P$  that sees all points 
of $S$ is path connected but there are points in $V(\{1,2,3,4\})$ that 
do not see each other.  Additionally, the set~$V(\{1,2,4\})$ is not 
path-connected.}
\label{notvisible-fig}
\end{center}
\end{figure*}

We give a precise outline of the proof.
\begin{enumerate}

\item The first observation is that $V(S)$ has to be path-connected. 
This is shown in Lemma~\ref{pathconnected-lem}.
Note that two points in $V(S)$ are not necessarily $L_1$-visible; see Figure~\ref{notvisible-fig}.  

\item The second observation is that for $p_i$, $i=1,\ldots,n$, 
any  $V(S\setminus\!\!\{p_i\})$ inside $P$ (brown areas within our examples) has to share a common boundary with $V(S)$ (the red area).  
This is shown in Lemma~\ref{Boundary-lem}. 
Additionally, the common boundary between any component of 
$V(S\setminus\{p_i\})$ and $V(S)$ stems from a $L_1$-cut labeled with a well-specified 
direction from $\{N,E,S,W\}$.  Note that  $V(S\setminus\{p_i\})$ need 
not be path-connected and can be separated from $V(S)$ in more than one direction. 

\item The third observation is that, if there are two points 
$p_i$ and $p_j$ such that components  $V(S\setminus\{p_i\})$ and 
$V(S\setminus\{p_j\})$ are separated from $V(S)$ 
by $L_1$-cuts of the same direction $X\in\{N,E,S,W\}$,
there can only be a single $L_1$-cut in direction$~X$ 
that contributes to the boundary of $V(S)$ and this cut separates both 
components of $V(S\setminus\{p_j\})$  and $V(S\setminus\{p_i\})$ from $V(S)$. 
The single $L_1$-cut $c$ is evoked by a vertex (or edge) $v_c$ that sees 
both points $p_i$ and $p_j$. This is shown in Lemma~\ref{TwoPointsDirection-lem}. 

\item A direct consequence is the following. The above mentioned
$L_1$-cut $c$ evoked by a vertex (or edge) $v_c$ separates $P$ into three disjoint parts, 
one of which, say $P_{v_c}(V(S))$, contains $V(S)$ and the other two, say $P_{v_c}(V(S\setminus\{p_j\}))$ and $P_{v_c}(S\setminus\{p_i\}))$, contain 
$p_i$ and $p_j$, respectively. 
Additionally, there is no component of a 
$V(S\setminus\{p_k\})$ for  $k\neq i,j$
 that can be separated from $V(S)$ by an $L_1$-cut into  
 direction~$X$. 
 

\item Since we only have four different directions, 
starting from $V(S)$, by the above argument we can have at most two points separated 
by a cut in the corresponding direction. Or in other words, 
there are at most 8 different sets $V(S\setminus\!\!\{p_i\})$ which can share the boundary with $V(S)$. 
This already means that the VC-dimension can be at most~8 which is the number of 
subsets of size 7 for a set of 8 points. 
 This is the statement of Corollary~\ref{SimpleConsec-cor}. 
 
\item Finally, 
we  have to do some investigation on the relative position of the 
aforementioned maximal 4 $L_1$-cuts. 
We show that for a fixed combination of a 
horizontal and vertical $L_1$-cut at most
three sets  $V(S\setminus\!\!\{p_i\})$ can be separated. 
Then we argue, that this combination cannot happen again 
in the opposite corner
so that in total only 5 points can survive.  This is shown in the 
proof of the final Theorem~\ref{5Points-theo}. Note that Figure~\ref{lowerbound-fig} exactly matches the worst-case situation.
\end{enumerate}

In the following we will always assume that 
the set $S=\{p_1,\ldots,p_n\}$ is shattered 
by visibility polygons inside a simple polygon $P$ and that 
 $V(T)$ for $T\subseteq S$ is the union of points 
 in $P$ that sees any point in $T$ but no point in $S\setminus T$ as defined as before. It will be explicitly mentioned, if it is necessary 
 to use $L_1$-visibility. 
 Furthermore, w.r.t. the notion of path, path-connected and 
 shortest path the next two Lemmata are actually independent  from the choice 
 of the metric  ($L_1$- or $L_2$), because we
 make use of short-cuts along a line segment that breaks the visibility, only. 
 For convenience let us assume that we consider $L_2$-paths but visibility 
 might be different. 

\begin{lemma}\label{pathconnected-lem} 
The subset $V(S)$ of $P$ is path-connected. 
\end{lemma}
\begin{proof}
Assume that two points $p,q\in V(S)$ are not connected by 
a path that fully runs inside $V(S)$. This means that along a shortest path 
$\mbox{SP}_P(p,q)$ between $p$ and $q$ inside $P$ there 
will be some point $q_1$ where some $p_i$ is not seen 
 \emph{after} $q_1$ for a while and comes into sight again at some point $q_2$ on $\mbox{SP}_P(p,q)$; compare the sketch in  Figure~\ref{PathConnected-fig}. 
\begin{figure*}
\begin{center}
\includegraphics[scale=0.45]{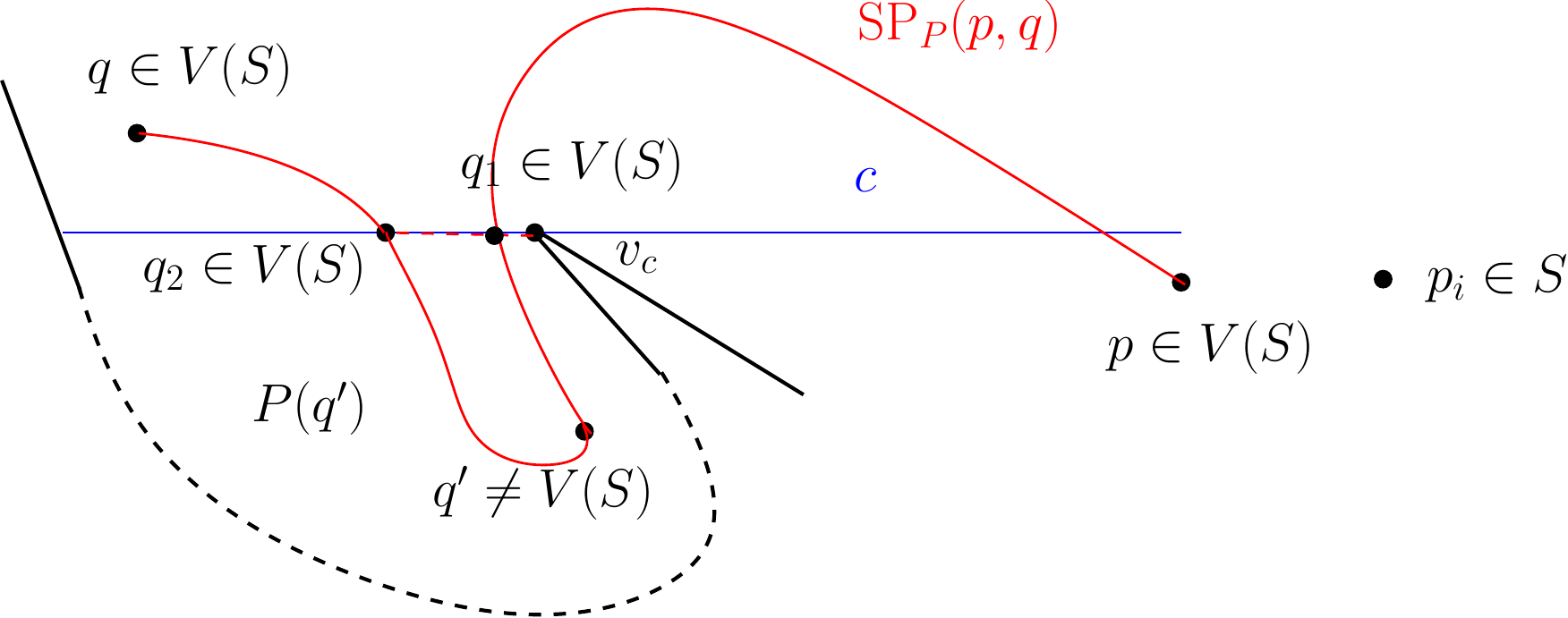}
\caption{We consider a shortest path $\mbox{SP}_P(p,q)$ between two points 
$p,q\in V(S)$. If there is some point $q'$ along the path with $q'\neq V(S)$ there 
is a cut $c$ such that a short-cut for  $\mbox{SP}_P(p,q)$ exist. 
There is always a shortest path between $p$ and $q$ that runs fully inside $V(S)$ and 
$V(S)$ is path-connected.}
\label{PathConnected-fig}
\end{center}
\end{figure*}
More precisely the path from $p$ to $q_1$ will cross some cut $c$ that emanates 
from some vertex $v_c$ (or edge $e_c$). The cut
$c$ blocks the visibility 
of some point $p_i\in S$ for some point $q'$ on $\mbox{SP}_P(p,q)$.
The cut subdivides $P$ 
into a part $P(q')$ that contains $q'$ and a part $P(p,q)$ 
that contains $p$ and $q$. Finally at point $q_2$ the 
path $\mbox{SP}_P(p,q)$ has to cross $c$ and enter 
$P(p,q)$ again in order to let $p_i$ become visible. 
Therefore along $c$ there will be a short-cut along $c$ 
using the segment between $q_1$ and $q_2$. There is a shortest path $\mbox{SP}_P(p,q)$ between $p$ and $q$ that runs fully inside $V(S)$, $V(S)$ is path-connected. \hfill$\qed$
\end{proof}

Note that the above arguments are independent from 
considering $L_1$- or $L_2$-visibility. 
Only  the fact that the 
cut is a line segment and allows a short-cut is used. 

\begin{lemma}\label{Boundary-lem}
Any set $V(S\setminus\!\!\{p_k\})$ of $P$ shares a common boundary with $V(S)$. 
\end{lemma}
\begin{proof} Let us assume that 
$V(S\setminus\!\!\{p_k\})$  and $V(S)$ do not share a common boundary. Thus, for \emph{any} two points $p\in V(S)$ and 
$q\in V(S\setminus\!\!\{p_k\})$ a shortest path $\mbox{SP}_P(p,q)$ 
between $p$ and $q$ in $P$ will leave $V(S)$ 
at some point $q_1$ to enter some $V(S')$ with $S'\neq S\setminus\!\!\{p_k\}$ and finally has to end in $V(S\setminus\!\!\{p_k\})$ at $q$.
This means at $q_1$ at least a point $p_i\in S$ with $p_i\neq p_k$ gets out of sight. 
With similar arguments as in the previous proof, the path 
$\mbox{SP}_P(p,q)$ has to cross some cut $c$ at $q_1$ 
and also at some point $q_2$ again in order to see  
$p_i\in S$ again. 
Therefore again we can short-cut $\mbox{SP}_P(p,q)$ 
by using the direct path between $q_1$ and $q_2$, which 
contradicts the assumption that no shortest path between 
$p\in V(S)$ and 
$q\in V(S\setminus\{p_k\})$  runs in $V(S)\cup V(S\setminus\!\!\{p_k\})$.
The sets $V(S\setminus\!\!\{p_k\})$ and~$V(S)$ share a 
common boundary. \hfill$\qed$
\end{proof} 

Note, that the arguments are again independent from 
$L_1$- and $L_2$-visibility.  
Again only the fact that the cut is a line segment that
allows a short-cut is used.

We make use of  $L_1$-visibility right now. 
The above Lemma says that $V(S)$ and any path-connected component 
(maximally path-connected subset) of 
$V(S\setminus\{p_i\})$ share a common edge. 
Obviously this edge has to stem from an $L_1$-cut that blocks the visibility to~$p_i$. Each such cut is labelled by a corresponding direction $\{N,E,S,W\}$ w.r.t.  the
relative position of its generating vertex (or edge).
For example the cut $c$ of vertex $v_c$ in Figure~\ref{PathConnected-fig} is labeled 
by direction $S$ (south). 
In the following for convenience by $V(S\setminus\!\!\{p_i\})$
we denote a path-connected component of the set $V(S\setminus\!\!\{p_i\})$. 

At this point we would like to mention the general position assumption. 
Under general position assumption we have uniqueness of the cuts and the 
corresponding vertices. Note that our arguments can be maintained for 
non-general position assumption as well. 
If there is more than one vertex (or edge) that defines the same $L_1$-cut because the vertices  (or edges) have the same height or 
width, we can make use of a unique vertex or edge that is 
responsible for the $L_1$-cut making the cut and its vertex unique. 
But it is allowed to change this vertex, if this is necessary. 

We now show that w.r.t. a specified direction at most two sets $V(S\setminus\!\{p_i\})$ and  $V(S\setminus\!\{p_j\})$ can be separated from $V(S)$. 
A corresponding  $L_1$-cut 
$c(p_i)$  associated to a vertex $v(p_i)$ subdivides the polygon into three parts, where one part, denoted by  $P_{v(p_i)}(V(S\setminus\!\!\{p_i\}))$, contains the corresponding portion of $V(S\setminus\!\!\{p_i\})$. 

\begin{lemma}\label{TwoPointsDirection-lem}
  If there are two  sets $V(S\setminus\!\{p_i\})$  and $V(S\setminus\!\{p_j\})$ that share 
  a common boundary with $V(S)$ evoked by  $L_1$-cuts
  in the same direction~$X$, there is only a unique, single $L_1$-cut $c(p_i,p_j)$ in direction~$X$ that shares the boundary between $V(S)$ and  both sets $V(S\setminus\!\{p_i\})$ and  $V(S\setminus\!\{p_j\})$. The sets $V(S\setminus\!\{p_i\})$ and  $V(S\setminus\!\{p_j\})$ lie to the 
left and right of the associated vertex (or edge) $v(p_i,p_j)$.  
The points $p_i$ and $p_j$ are $L_1$-visible from $v(p_i,p_j)$.
\end{lemma}
\begin{proof}

Assume that  two sets $V(S\setminus\!\!\{p_i\})$ and  $V(S\setminus\!\!\{p_j\})$ 
are connected to $V(S)$ by portions of \emph{different} $L_1$-cuts $c(p_i)$ 
and $c(p_j)$ of the same direction~$X$. Let $v(p_i)$  and $v(p_j)$ 
denote the (unique) vertices (or edges) that evoke $c(p_i)$ 
and $c(p_j)$. 
W.l.o.g. we assume that $X=S$ holds 
and  $c(p_i)$ and $c(p_j)$ are therefore horizontal cuts. 

Since  some point on $c(p_i)$ on the boundary of $P_{v(p_i)}(V(S\setminus\!\!\{p_i\}))$ lies in $V(S)$ and sees $p_i$, 
$p_i$ is $L_1$-visible from $v(p_i)$. Analogously, $p_j$ is $L_1$-visible from $v(p_j)$.
By general position assumption $c(p_i)$ and $c(p_j)$  do not have the same height. W.l.o.g. let the $Y$-coordinate $v(p_i)$  be larger than the $Y$-coordinate of $v(p_j)$, the other case is symmetric. 

Relative to the unique vertices (or edges)  $v(p_i)$  and $v(p_j)$ that evokes 
$c(p_i)$ and $c(p_j)$, the points $p_i$ or $p_j$ lie to the \emph{left} 
or \emph{right} from $v(p_i)$  or $v(p_j)$, meaning that 
$P_{v(p_i)}(V(S\setminus\!\!\{p_i\}))$ or $P_{v(p_j)}(V(S\setminus\!\!\{p_j\}))$ (the
caves containing $V(S\setminus\!\!\{p_i\})$  or $V(S\setminus\!\!\{p_j\})$, respectively) is on the opposite side; see Figure~\ref{OneDirectionB-fig}.
\begin{figure}
\begin{center}
\includegraphics[scale=0.5,viewport = 200 0 300 200]{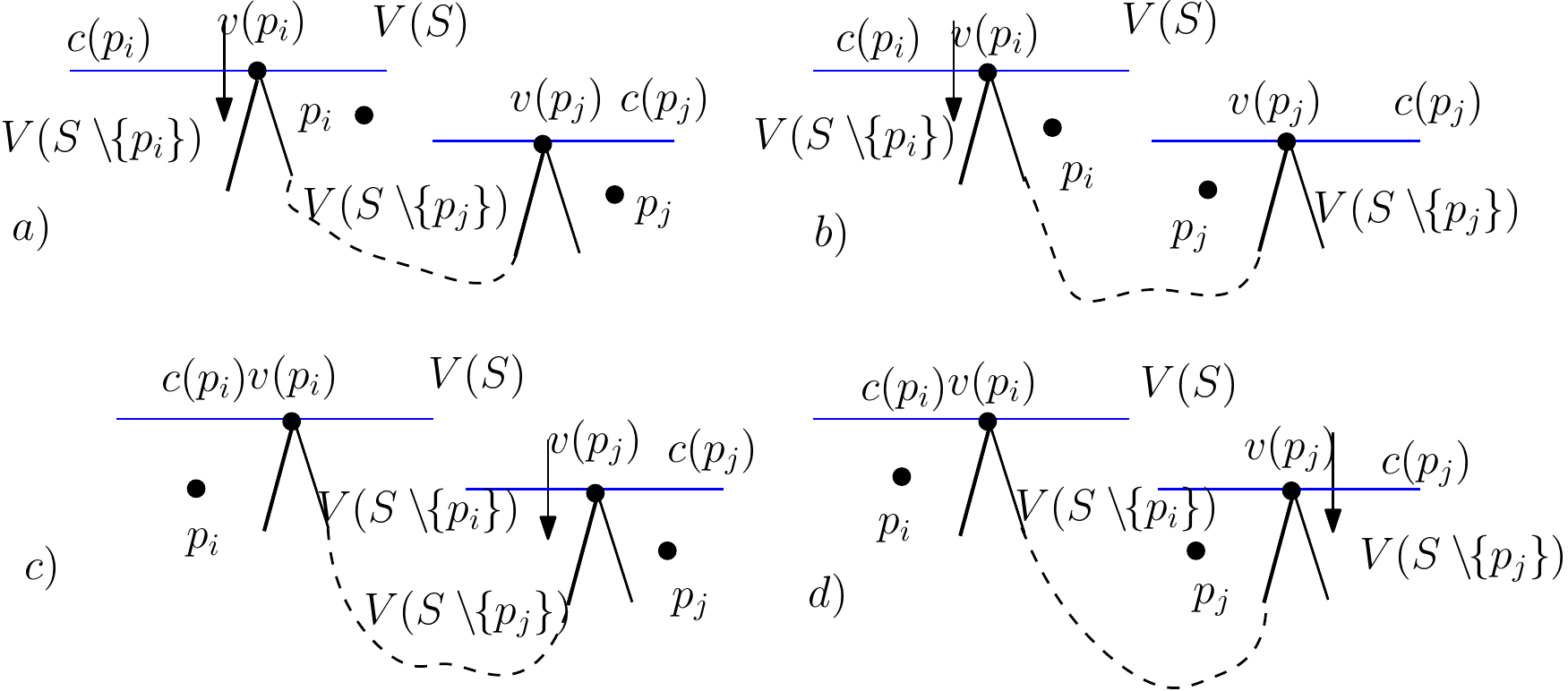}
\caption{Two sets $V(S\setminus\!\{p_i\})$ and  $V(S\setminus\!\{p_j\})$ that
are connected to $V(S)$ by some $L_1$-cuts of the same direction~$X$, 
have to be connected to $V(S)$ by a unique common $L_1$-cut. 
If two different $L_1$-cuts of the same direction have different height, 
there is always one $L_1$-cut that separates both $p_i$ and $p_j$, a contradiction.
Up to mirroring, the cases a)-d) contain the relative positions of $p_i$ and $p_j$ w.r.t. the vertices $v(p_i)$ and 
$v(p_j)$. None of the cases can occur, a single cut has to be responsible for both point sets $V(S\setminus\!\!\{p_i\})$ and  $V(S\setminus\!\!\{p_j\})$.
}
\label{OneDirectionB-fig}
\end{center}
\end{figure}
Up to mirroring, we now consider  all possible situations of the relative position of $p_i$ and $p_j$ w.r.t. 
$v(p_i)$ and $v(p_j)$; compare Figure~\ref{OneDirectionB-fig}~a)-d). 
In any case (indicated by an arrow) at least one 
of the cuts breaks the $L_1$-visibility to both points $p_i$ and $p_j$, which contradicts the assumption 
that the corresponding cut separates only a single point.  
In any case we have a contradiction to the assumption.\hfill$\qed$

\end{proof}

The above Lemma already implies that
the number of different sets~$V(S\setminus\!\{p_i\})$ 
 that share a common boundary with $V(S)$ for a fixed direction~$X$ can be at most~$2$. Any such pair $(p_i,p_j)$ is 
 separated by a single, unique $L_1$-cut $c$ evoked by some vertex (or edges) $v_c$. Two different such pairs for one direction cannot exist. 
 For such a unique $L_1$-cut $c$ there will be one point $p_i$ 
to the left of $v_{c}$ and another point $p_j$ to the right of  $v_{c}$, 
Both points are $L_1$-visible from $v_{c}$. 
For shattering at least $n$ points we require at least $n=\binom{n}{n-1}$ subsets $V(S\setminus\{p_i\})$ around $V(S)$. 

\begin{corollary}\label{SimpleConsec-cor}
For any direction~$X$ at most two sets $V(S\setminus\!\!\{p_i\})$ can share the boundary with $V(S)$. 
The VC-dimension for $L_1$-visibility w.r.t. points in simple polygons is not larger than~$8$. 
\end{corollary}

Now assume that we have a maximum number of  sets $V(S\setminus\!\!\{p_i\})$ 
located around $V(S)$, which is~$8$ in total. 
In this case in any direction  two different points $p_i$ and $p_j$ such that 
 $V(S\setminus\!\!\{p_i\})$ and $V(S\setminus\!\!\{p_j\})$ are separated by a single $L_1$-cut of direction~$X$; the situation is sketched in Figure~\ref{EightPairs-fig}. 
\begin{figure}
\begin{center}
\includegraphics[scale=0.5]{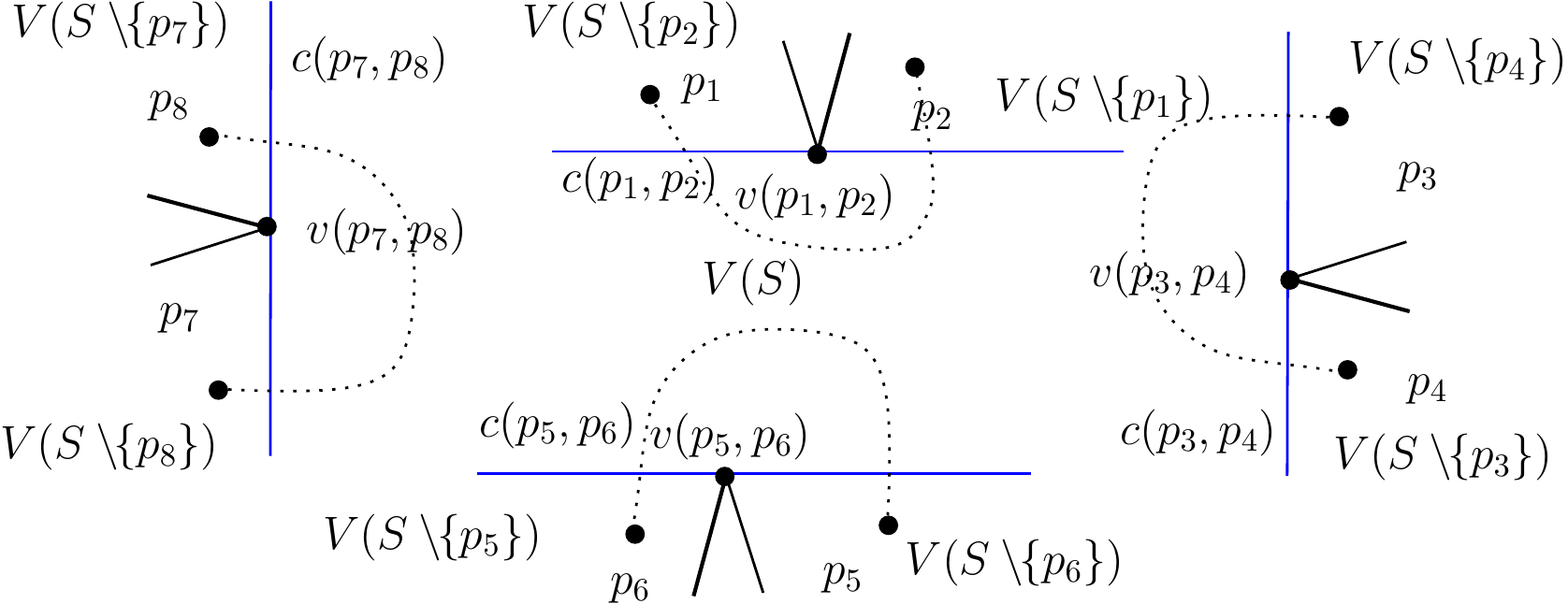}
\caption{The sketch of the worst-case situation where in any direction two different 
sets $V(S\setminus\!\!\{p_i\})$ and $V(S\setminus\!\!\{p_j\})$ build a boundary 
with $V(S)$. This situation will be considered more precisely, 
only~$5$ points can survive. 
}
\label{EightPairs-fig}
\end{center}
\end{figure}

 Finally, we consider the situation of a horizontal \emph{and} a vertical $L_1$-cut 
and the corresponding potential sets $V(S\setminus\!\!\{p_i\})$. 
We argue that at most three sets $V(S\setminus\!\!\{p_i\})$ 
 can be attained. Considering two corners or \emph{all} cuts only~$5$ points survive.  

\begin{theorem}\label{5Points-theo}
The VC-dimension for $L_1$-visibility w.r.t. points in simple polygons is exactly~$5$. 
\end{theorem}
\begin{proof}
The proof works as follows. 
Starting from at most $4$~pairs of potential points  for~$4$ directions as
sketched in Figure~\ref{EightPairs-fig} we  first consider the combination of a horizontal and a vertical cut and the maximum number of sets 
$V(S\setminus\!\{p_i\})$ that can be attained. 
It turns out that for such a single \emph{corner situation} 
either three sets  or two sets $V(S\setminus\!\{p_i\})$ can be constructed
depending on the constitution of the cuts; see Figure~\ref{CaseComplete1-fig}  
and Figure~\ref{CaseComplete23-fig}. 
Then the final situation consists of two opposite corners. 
We show that the configuration for three sets $V(S\setminus\!\!\{p_i\})$ 
cannot happen for two opposite corners; see  Figure~\ref{AppAbsoluteFinal-fig}.  Therefore in total at most 5 sets $V(S\setminus\!\!\{p_i\})$  can be attained. Indeed a combination of Case 1 of Figure~\ref{CaseComplete1-fig} (three sets)  and 
Case 3 of Figure~\ref{CaseComplete23-fig} (two sets) in  the opposite 
corners gives the lowerbiund bound of Figure~\ref{lowerbound-fig}.



Now as mentioned above consider the combination of a horizontal and a vertical cut.
Assume that both cuts contribute to the boundary of $V(S)$. 
If this is not the case, we would have even less sets $V(S\setminus\!\!\{p_i\})$.
Let us first present the final result for the two cuts, w.l.o.g. a horizontal cut 
of direction~$N$ and 
a vertical cut of direction~$W$. How many sets $V(S\setminus\!\!\{p_i\})$ can be constituted? 
Depending on the position of the evoking vertices
 and up to symmetry  (related to this \emph{corner}) only the three cases as depicted in Figure~\ref{CaseComplete1-fig} (3 sets $V(S\setminus\!\!\{p_i\})$) and Figure~\ref{CaseComplete23-fig}  (2 sets $V(S\setminus\!\!\{p_i\})$)
 can occur. For any additional other point  $p$ the set $V(S\setminus\!\!\{p\})$ has to be separated by a cut of a different direction (here~$S$ or~$E$).  
\begin{figure}
\begin{center}
\includegraphics[scale=0.51,page=1]{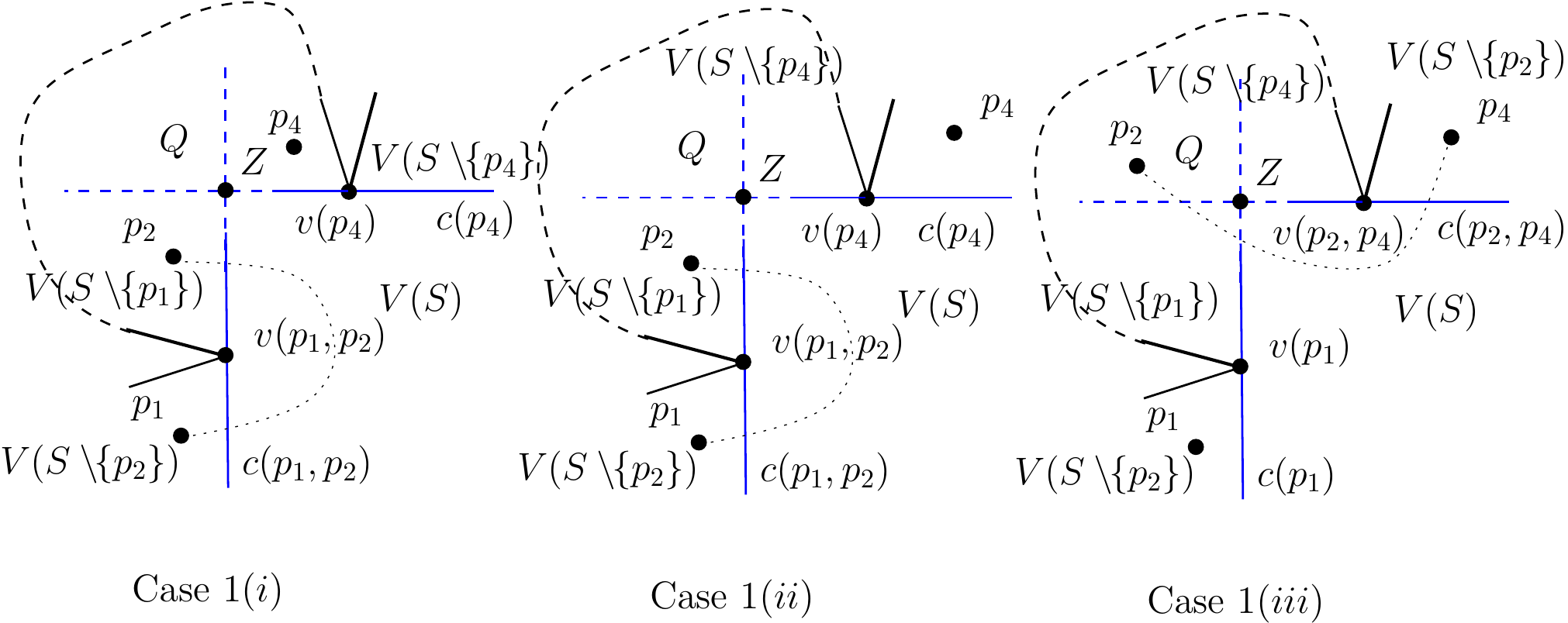}
\caption{
The intersection $Z$ build by (the extensions of) a horizontal and vertical cut  that share the boundary with~$V(S)$. If the  evoking vertices  lie to the right and 
below~$Z$, we can obtain at most three sets $V(S\setminus\!\!\{p_i\})$. 
For this corner up to symmetry the cases 1(i), 1(ii) and 1(iii) can occur. }
\label{CaseComplete1-fig}
\end{center}
\end{figure}
\begin{figure}
\begin{center}
\includegraphics[scale=0.51,page=1]{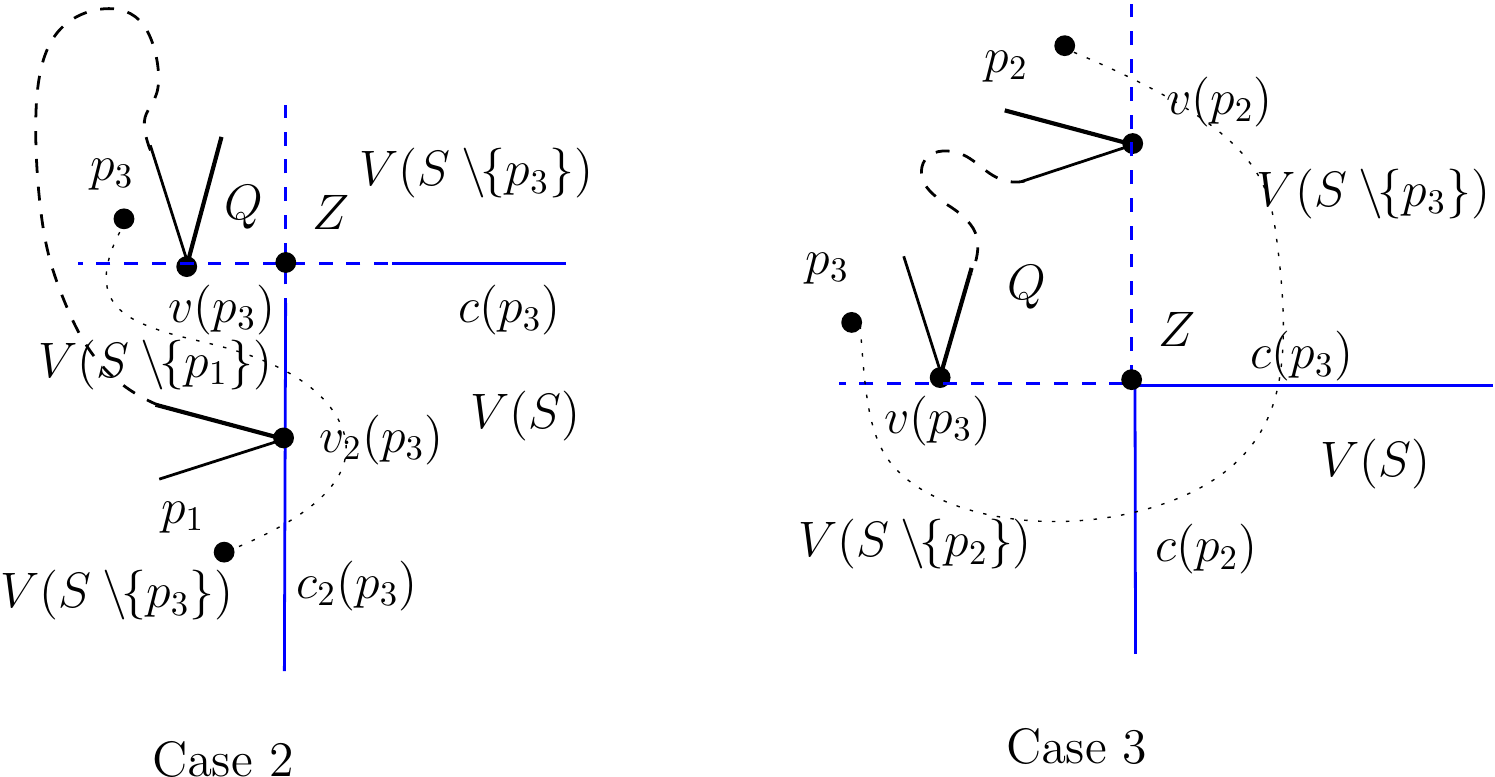}
\caption{
The intersection $Z$ build by (the extensions of) a horizontal and vertical cut 
 that share the boundary with~$V(S)$. If the evoking vertices 
 of the cuts lie to the left and below~$Z$ (Case 2) or to 
 the left and above~$Z$ (Case 3), in this corner up to symmetry we can obtain 
 at most two sets $V(S\setminus\!\!\{p_i\})$. }
\label{CaseComplete23-fig}
\end{center}
\end{figure}

For the proof that up to symmetry (related to the corner) only the cases of Figure~\ref{CaseComplete1-fig} (3 sets $V(S\setminus\!\!\{p_i\})$) and Figure~\ref{CaseComplete23-fig}  (2 sets $V(S\setminus\!\!\{p_i\})$)
 can occur, we consider a 
 potential pair $(p_1,p_2)$ and a corresponding vertical cut $c(p_1,p_2)$ evoked by 
vertex $v(p_1,p_2)$ of direction~$W$ and a potential
pair $(p_3,p_4)$ with a corresponding vertical cut $c(p_3,p_4)$ evoked by 
vertex $v(p_3,p_4)$ of direction~$N$.  The meaning is that we would like 
to find out how many different  sets $V(S\setminus\!\!\{p_i\})$ can be separated  at most from $V(S)$ by the given cuts.   
Note that the corner case for~$S$ and~$E$ is symmetric.

 Now, we consider the intersection point~$Z$ of the two lines passing through  $c(p_1,p_2)$ 
and $c(p_3,p_4)$. 
The vertices $v(p_1,p_2)$ and $v(p_3,p_4)$ have a relative position with respect to the intersection point~$Z$. In this corner by symmetry only three cases have 
to be considered.
The upper left axis-parallel quadrant of origin $Z$ is 
denoted by~$Q$.   
\pagebreak
\begin{enumerate}
\item $v(p_1,p_2)$ lies below~$Z$ and 
$v(p_3,p_4)$ lies to the right of~$Z$; see Figure~\ref{AppCase1abVertHoriz-fig}.
\begin{enumerate}

\item Either $p_2$ and/or $p_3$ lies inside $Q$; see 
Figure~\ref{AppCase1abVertHoriz-fig}~1(a).
\item Neither  $p_2$ nor $p_3$ lies inside $Q$; see Figure~\ref{AppCase1abVertHoriz-fig}~1(b).
\end{enumerate}
\item $v(p_1,p_2)$ lies below~$Z$ and 
$v(p_3,p_4)$ lies to the left  of~$Z$; see Figure~\ref{AppCase23VertHoriz-fig}~2. 
\item $v(p_1,p_2)$ lies above~$Z$ and 
$v(p_3,p_4)$ lies to the left  of~$Z$; see  Figure~\ref{AppCase23VertHoriz-fig}~3.
\end{enumerate}

For Case 1 we have two sub-cases. For Case 1(a) 
let $p_3$ lie inside the quadrant $Q$ (upper-left quadrant from~$Z$),   then below $v(p_1,p_2)$ there is no region $V(S\setminus\!\!\{p_2\})$ connected to $V(S)$, 
so $p_2$ does not belong to the two cuts. 

Assume that both  sets $V(S\setminus\!\!\{p_2\})$ and 
$V(S\setminus\!\!\{p_3\})$ for direction~$W$ and~$N$ exist as depicted in Figure~\ref{AppCase1abVertHoriz-fig}~1(b).
The two sets $V(S\setminus\!\!\{p_2\})$ and 
$V(S\setminus\!\!\{p_3\})$  are well-separated from each other. 
Assume that  $p_1$ and $p_4$ exist  for the given cuts or more precisely $V(S\setminus\!\!\{p_1\})$ and $V(S\setminus\!\!\{p_4\})$ are separated by the cuts in direction~$N$ and~$W$, respectively. 
There will be no subsets $V(S\setminus\!\!\{p_2\})$ and $V(S\setminus\!\!\{p_3\})$ separated by cuts of direction~$S$ or~$E$, respectively. 
This holds because a corresponding cut of direction~$S$ has to run above $v(p_1,p_2)$ 
and also separates~$p_1$ and a corresponding cut of direction~$E$ has to run to the left of $v(p_3,p_4)$  and also separates~$p_4$. See for example  Figure~\ref{AppCase1abVertHoriz-fig}~1(b) for the point $p_2$. 
 
 Altogether, if $p_1$ and $p_4$ exist  for the given cuts and
 we would like to shatter $V(S\setminus\!\!\{p_2,p_3\})$ from some point in $P$, we have to 
enter one of the sets $V(S\setminus\!\!\{p_2\})$ or $V(S\setminus\!\!\{p_3\})$ 
 from $V(S)$ separated by the given cuts. Assume that we would like to shatter 
$V(S\setminus\!\!\{p_2,p_3\})$ and move inside $V(S\setminus\!\!\{p_2\})$, the other case is symmetric. If we would like to get 
$p_3$ out of sight, we will also loose visibility to $p_4$. So
$V(S\setminus\!\!\{p_2,p_3\})$  cannot be shattered, if both points $p_1$ and $p_4$ exist or if both points $p_2$ and $p_3$ exist. 
So in any combination at most three points can exist which results in Case 1(i), 
1(ii) or 1(iii) of Figure~\ref{CaseComplete1-fig} or its symmetric counterparts.
%
\begin{figure}
\begin{center}
\includegraphics[scale=0.48,page=1]{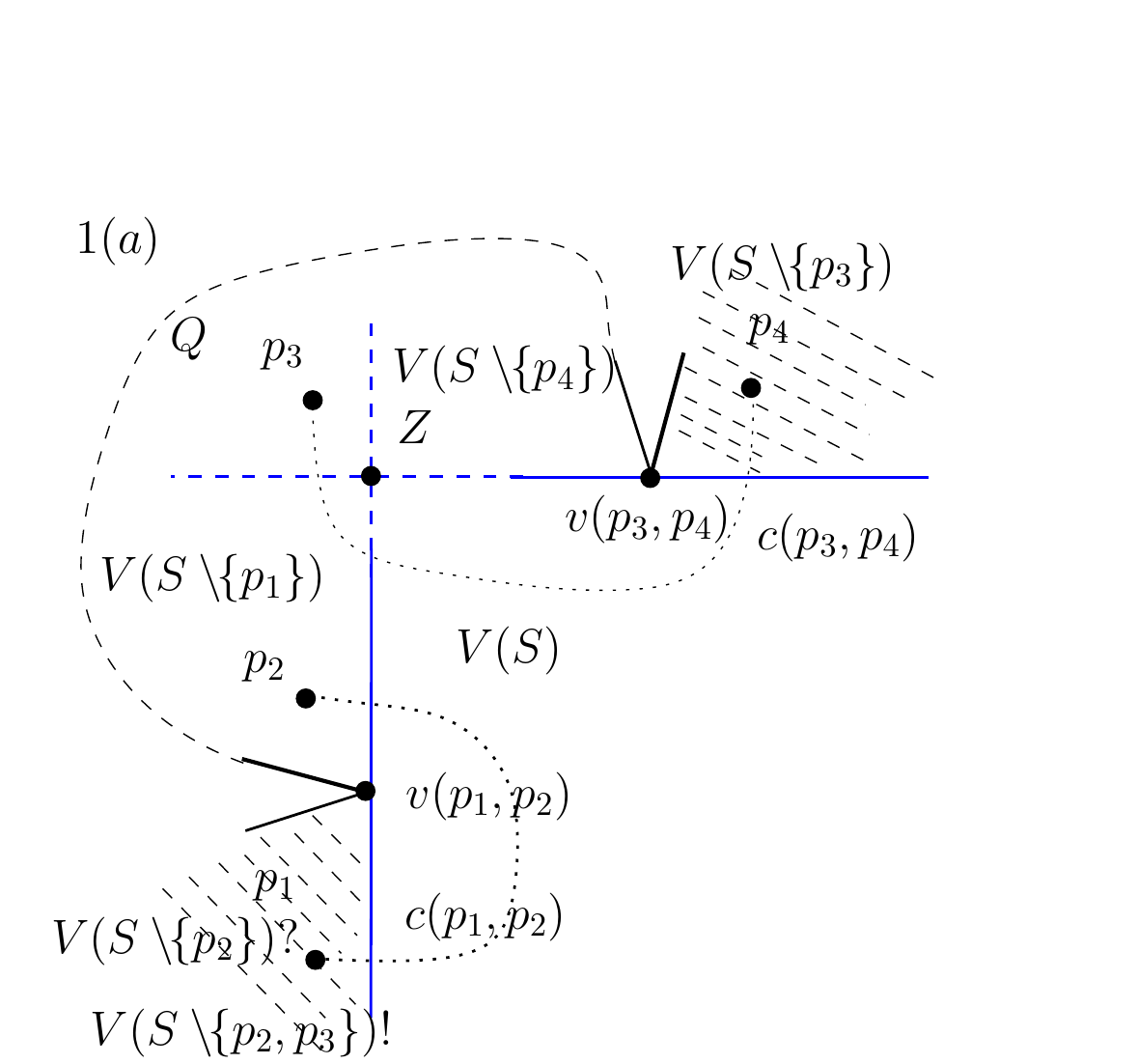}
\includegraphics[scale=0.48,page=1]{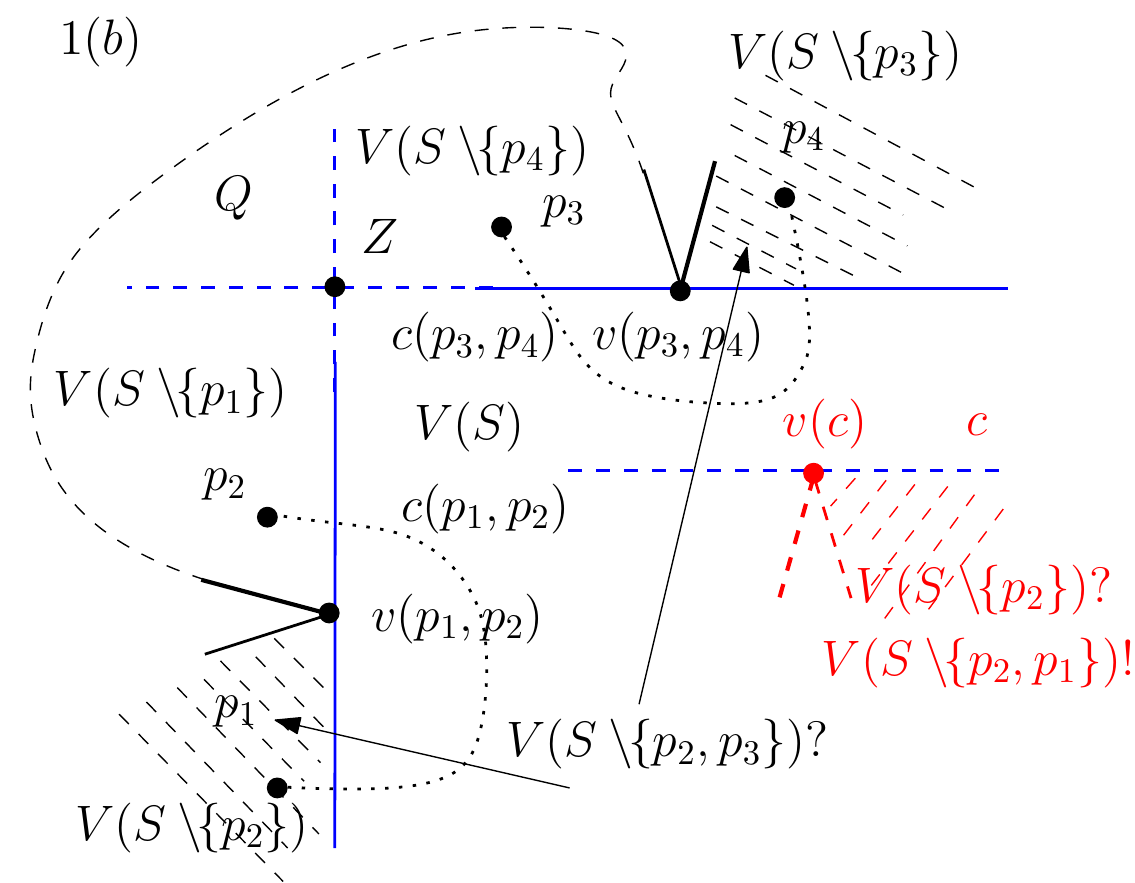}
\caption{If $p_3$ lies inside $Q$ the set $V(S\setminus\!\!\{p_2\})$ 
cannot be separated by the given  cuts. Thus for Case 1(a) either $p_2$ or $p_3$ does not exist.   For Case 1(b) if $p_1$ and $p_4$  exist, the set $V(S\setminus\!\!\{p_2\})$  cannot be 
separated by direction~$S$ and the set 
$V(S\setminus\!\!\{p_3\})$ 
cannot be separated by direction~$E$. Therefore we cannot shatter 
the set  $V(S\setminus\!\!\{p_2,p_3\})$ if $p_4$ and $p_1$ exist. 
This means that either one point from $p_1$ and $p_4$ does not exist or 
one point from  $p_2$ and $p_3$. } 
\label{AppCase1abVertHoriz-fig}
\end{center}
\end{figure}

In  Case 2 $v(p_3,p_4)$ lies to the left  of~$Z$ and $v(p_1,p_2)$ lies below~$Z$ as depicted in  Figure~\ref{AppCase23VertHoriz-fig}~2. 
 First, we notice that $V(S\setminus\!\!\{p_4\})$ has to be separated by
 direction~$E$, therefore we conclude that $p_4$ cannot belong to the
 given cuts.    Additionally,  w.r.t. the position of $p_3$ we have the same situation as given in Case 1(a) because 
 $p_3$ has to lie inside~$Q$. Similar to Figure~\ref{AppCase1abVertHoriz-fig}~1(a), the set $V(S\setminus\!\!\{p_2\})$ cannot be shattered by the given cuts. 
 At least one of the points $p_2$ or $p_3$ cannot exist. 
Note that if $p_3$ does not exist, the cut of direction~$N$ is not used at all. 
Since we would like to exploit both cuts only $p_3$ and $p_1$ remains. 
This results in Case 2 of Figure~\ref{CaseComplete23-fig}.
 %
\begin{figure}
\begin{center}
\includegraphics[scale=0.5,page=1]{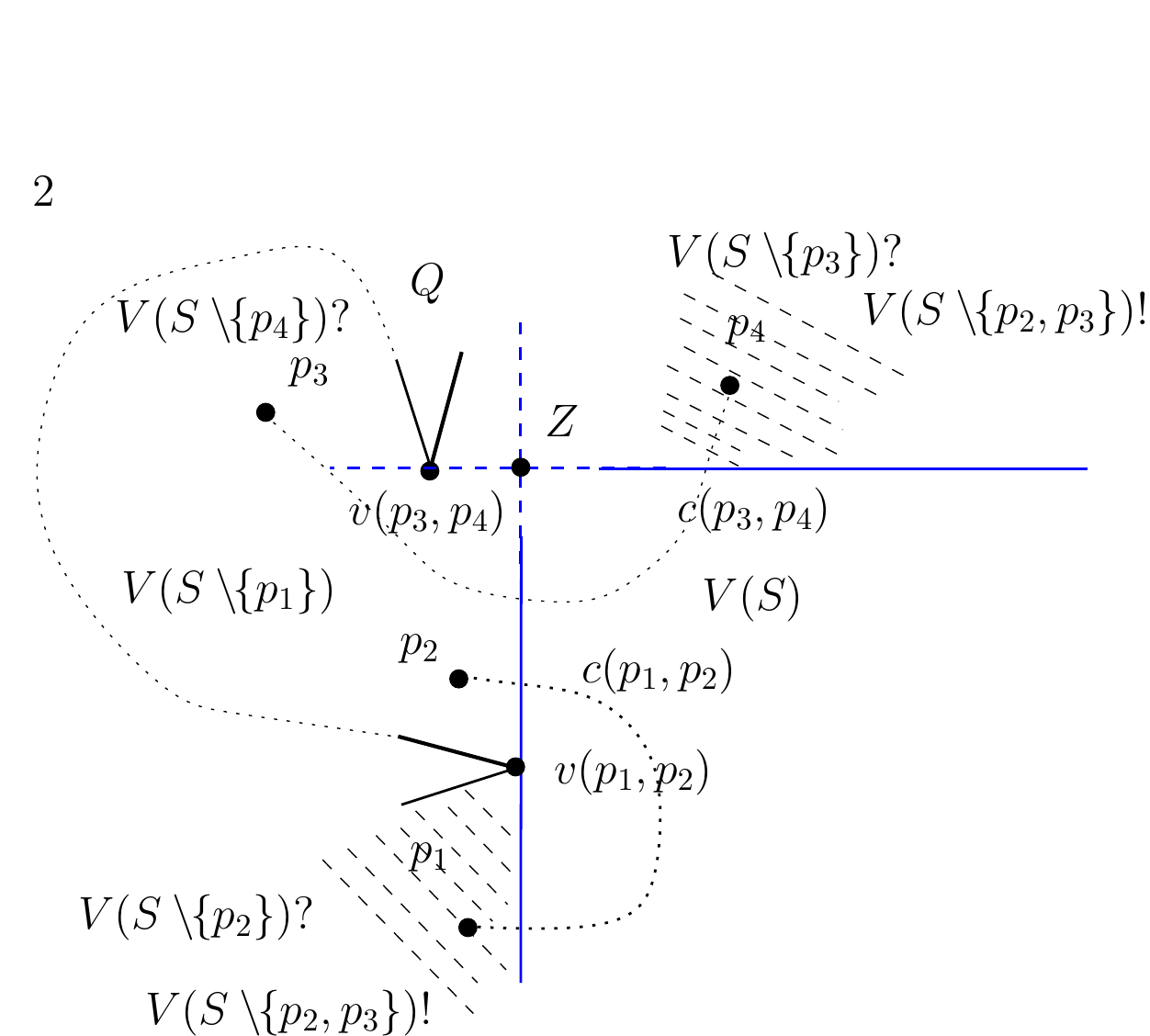}
\includegraphics[scale=0.5,page=1]{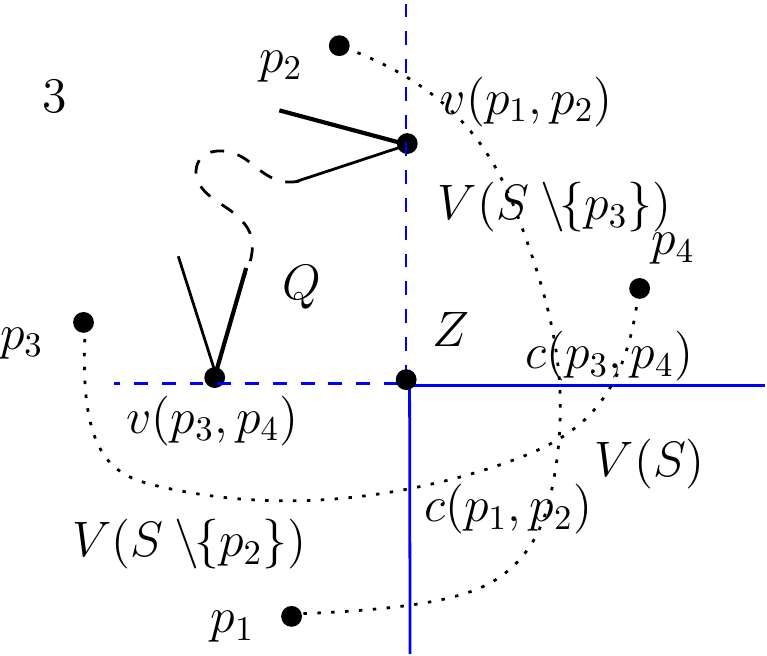}
\caption{For Case 2 the set $V(S\setminus\!\!\{p_2\})$ cannot be separated from $V(S)$ at all. The point $p_4$ or more precisely the set $V(S\setminus\!\!\{p_4\})$ does not belong to the cuts of the given directions and has to be omitted. 
Only $p_3$ exist because otherwise the cut $c(p_3,p_4)$ is useless. 
In Case 3, if $v(p_1,p_2)$ lies above~$Z$ and $v(p_3,p_4)$ lies to the left  of~$Z$ only the sets $V(S\setminus\!\!\{p_2\})$ and  $V(S\setminus\!\!\{p_3\})$ are separated by the given cuts. 
The points $p_1$ and $p_4$ can be omitted.
}
\label{AppCase23VertHoriz-fig}
\end{center}
\end{figure}

 In the remaining Case 3, $v(p_1,p_2)$ lies above~$Z$ and 
$v(p_3,p_4)$ lies to the left  of~$Z$ and we have a situation as given in 
Figure~\ref{AppCase23VertHoriz-fig}~3. Here $p_2$ has to be above $v(p_1,p_2)$ and $p_3$ lies to the left of $v(p_3,p_4)$. Additionally, $p_1$ has to be 
below  $v(p_1,p_2)$ and $p_4$ lies to the right of $v(p_3,p_4)$. 
The sets $V(S\setminus\!\!\{p_4\})$ and  $V(S\setminus\!\!\{p_1\})$
are not separated from the given cuts, $p_4$ and $p_1$ have to be omitted. 
This results in Case 3 of Figure~\ref{CaseComplete23-fig}.

Now for the final argumentation we have to combine the cases. 
Note that the combination of Case 3 and the application of a symmetric version of  Case 1 in the opposite corner results in our lower-bound construction. 
The above arguments already mean that we can shatter at most~$6$ points, if we 
apply Case 1 and its symmetric version for the opposite corner twice. This is the remaining case. 

Case 1 makes use of three points and allows that some $p_2'$ lies 
inside the given $Q$ as indicated 
by configuration $(p_1,p_2',p_3)$ in Figure~\ref{AppAbsoluteFinal-fig}.  
If this happens for the upper left corner, for shattering $6$ points in total we cannot apply 
Case 1 again to the opposite corner because for the upper right corner 
$Q'$ or for the lower left corner $Q''$ we would have a contradiction to Case 1; 
see Figure~\ref{AppAbsoluteFinal-fig}.  
\begin{figure}
\begin{center}
\includegraphics[scale=0.5,page=2]{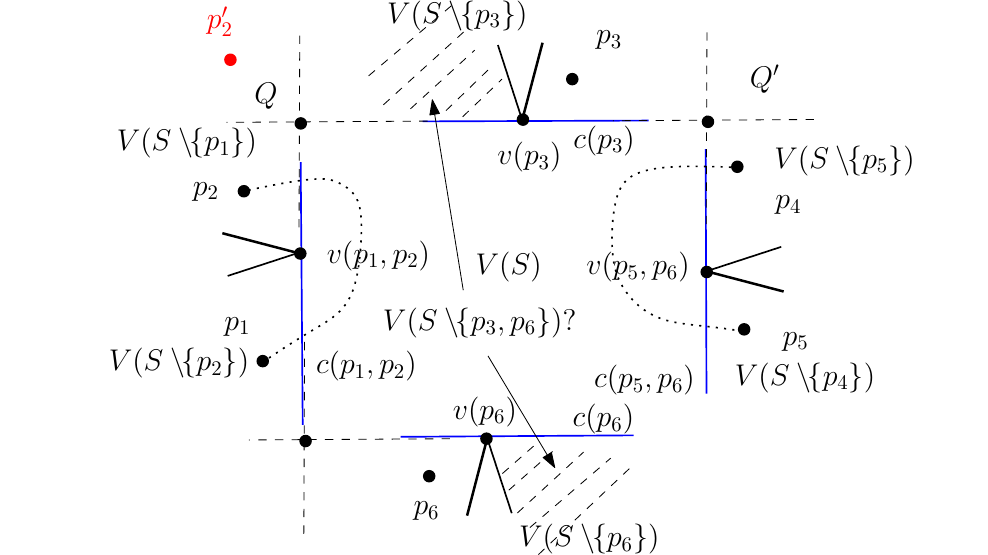}
\caption{The remaining case considers Case 1 twice. 
If $p_2'$ lies in $Q$, application of Case~1 in the opposite corner is not possible. 
So for shattering~$6$ points up to symmetry two 
sets $V(S\setminus\!\!\{p_j\})$ and $V(S\setminus\!\!\{p_i\})$ 
are separated exclusively by opposite directions. Here we have $p_i=p_6$ and 
$p_j=p_3$. The set $V(S\setminus\!\!\{p_6,p_3\})$ cannot be shattered.}
\label{AppAbsoluteFinal-fig}
\end{center}
\end{figure}

This means that we can have at 
most~6 points as depicted in Figure~\ref{AppAbsoluteFinal-fig} 
where for  two opposite directions in each direction two sets $V(S\setminus\!\!\{p_i\})$
  and $V(S\setminus\!\!\{p_j\})$ are separated from $V(S)$ and 
  for the two remaining opposite directions in each direction 
  only one set  $V(S\setminus\!\!\{p_j\})$ is separated from $V(S)$. 
  W.l.o.g. we choose direction~$E$ and~$W$ for the two sets that 
  are separared 
  and~$N$ and~$S$ for the remaining two sets, 
  say $V(S\setminus\!\!\{p_3\})$ and $V(S\setminus\!\!\{p_6\})$ as in 
  Figure~\ref{AppAbsoluteFinal-fig}. 
  Now we can argue that $V(S\setminus\!\!\{p_6,p_3\})$ cannot be 
    shattered. Starting from $V(S)$ we have to move inside $V(S\setminus\!\!\{p_6\})$  or $V(S\setminus\!\!\{p_3\})$. 
    If we would like to loose the visibility to the corresponding opposite point, we definitely also loose visibility to some point on the 
    remaining two directions.
    
    Altogether, we cannot apply Case 1 twice, only~$5$ points can be shattered by $L_1$-visibility polygons. \hfill$\qed$ 
\end{proof}

\vspace*{-0.4cm}
\section{Conclusion}\label{concl-sect}

We have shown that the VC-dimension for $L_1$-visibility of points in simple polygons is exactly~$5$. This result holds for any area 
that is enclosed by a simple Jordan curve. 
The VC-dimension plays an important role for the number of 
guards required for art gallery problems. 
Our prove idea mainly considers the relative position of the 
sets $V(T)\in P$ that sees exactly the subsets $T=S$, $T=S\setminus\!\!\{p_i\}$ and 
$T=S\setminus\!\!\{p_i,p_j\}$. Therefore we even show a slightly stronger result, because shattering these sets can only be 
done for exactly~$5$ points. 
The main open question is, whether we can exploit such properties 
for better upperbounds for the $L_2$-visibility case. 
Figure~\ref{lowerboundL2-fig} shows the best known lower bound for 
$L_2$-visibility.  
%
\begin{figure}
\begin{center}
\includegraphics[scale=0.35,viewport = 150 35 450 450]{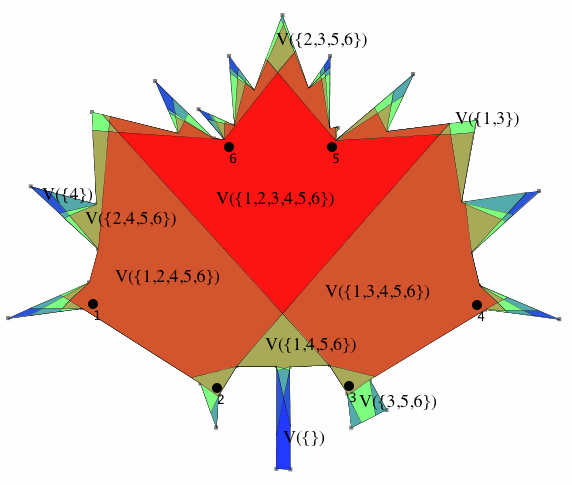}
\caption{The lower bound construction for the VC-dimension 
of points for $L_2$-visibility in simple polygons from Valtr~\cite{v-ggwnp-98}.
All $2^6$ subsets can be shattered, some regions are labeled by the 
point sets that are visible. 
}
\label{lowerboundL2-fig}
\end{center}
\end{figure}

 {
\bibliographystyle{abbrv}
\bibliography{%
        VCDIM}

\begin{thebibliography}{10}

\bibitem{g-aaagp-10}
S.~K. Ghosh.
\newblock Approximation algorithms for art gallery problems in polygons.
\newblock {\em Discrete Applied Mathematics}, 158(6):718 -- 722, 2010.

\bibitem{gkw-vvmp-14}
M.~Gibson, E.~Krohn, and Q.~Wang.
\newblock On the vc-dimension of visibility in monotone polygons.
\newblock In {\em Proceedings of the 26th Canadian Conference on Computational
  Geometry, {CCCG} 2014, Halifax, Nova Scotia, Canada, 2014}, 2014.

\bibitem{gk-nrvsp-09}
A.~Gilbers and R.~Klein.
\newblock New results on visibility in simple polygons.
\newblock In {\em Proceedings of 11th International Symposium on Algorithms and
  Data Structures, WADS}, pages 327--338, 2009.

\bibitem{gk-nubvv-14}
A.~Gilbers and R.~Klein.
\newblock A new upper bound for the vc-dimension of visibility regions.
\newblock {\em Comput. Geom.}, 47(1):61--74, 2014.

\bibitem{ikdv-vev-04}
V.~Isler, S.~Kannan, K.~Daniilidis, and P.~Valtr.
\newblock Vc-dimension of exterior visibility.
\newblock {\em {IEEE} Trans. Pattern Anal. Mach. Intell.}, 26(5):667--671,
  2004.

\bibitem{km-ggwep-97}
G.~Kalai and J.~Matou{\v{s}}ek.
\newblock Guarding galleries where every point sees a large area.
\newblock {\em Israel Journal of Mathematics}, 101(1):125--139, 1997.

\bibitem{kv-iclt-94}
M.~J. Kearns and U.~V. Vazirani.
\newblock {\em An Introduction to Computational Learning Theory}.
\newblock MIT Press, Cambridge, MA, USA, 1994.

\bibitem{kk-iagsg-11}
J.~King and D.~G. Kirkpatrick.
\newblock Improved approximation for guarding simple galleries from the
  perimeter.
\newblock {\em Discrete {\&} Computational Geometry}, 46(2):252--269, 2011.

\bibitem{k-ggh-00}
D.~G. Kirkpatrick.
\newblock Guarding galleries with no hooks.
\newblock In {\em Proceedings of the 12th Canadian Conference on Computational
  Geometry, Fredericton, New Brunswick, Canada, August 16-19, 2000}, pages
  43--46, 2000.

\bibitem{kpw-atbfe-92}
J.~Koml{\'o}s, J.~Pach, and G.~Woeginger.
\newblock Almost tight bounds for $\epsilon$-nets.
\newblock {\em Discrete {\&} Computational Geometry}, 7(2):163--173, 1992.

\bibitem{m-ldg-02}
J.~Matousek.
\newblock {\em Lectures on Discrete Geometry}.
\newblock Springer-Verlag New York, Inc., Secaucus, NJ, USA, 2002.

\bibitem{o-agta-87}
J.~O'Rourke.
\newblock {\em Art Gallery Theorems and Algorithms}.
\newblock Oxford University Press, Inc., New York, NY, USA, 1987.

\bibitem{ps-vgp-15}
A.~Padalkin and T.~Scheurich.
\newblock Computational geometry project: Visualising vc-dimension in polygons.
\newblock University of Bonn, Software Project supervised by Elmar Langetepe.
  2015.

\bibitem{v-ggwnp-98}
P.~Valtr.
\newblock Guarding galleries where no point sees a small area.
\newblock {\em Israel Journal of Mathematics}, 104(1):1--16, 1998.

\bibitem{vc-ucrfe-71}
V.~N. Vapnik and A.~Y. Chervonenkis.
\newblock On the uniform convergence of relative frequencies of events to their
  probabilities.
\newblock {\em Theory of Probability and its Applications}, 16(2):264--280,
  1971.

\end{thebibliography}
        }

%




\end{document}